\documentclass{article}

\usepackage{arxiv}

\usepackage[utf8]{inputenc} 
\usepackage[T1]{fontenc}    
\usepackage{hyperref}       
\usepackage{url}            
\usepackage{booktabs}       
\usepackage{amsfonts}       
\usepackage{nicefrac}       
\usepackage{microtype}      
\usepackage{lipsum}		
\usepackage{graphicx}
\usepackage{doi}
\usepackage{amsthm}
\theoremstyle{plain}
\newtheorem{theorem}{Theorem}

\newtheorem{corollary}{Corollary}
\newtheorem{proposition}[theorem]{Proposition}
\newtheorem{assumption}{Assumption}
\newtheorem{fact}{Fact}
\newtheorem{question}{Question}

\theoremstyle{definition}
\newtheorem{definition}{Definition}

\theoremstyle{remark}
\newtheorem{remark}{Remark}

\title{M\" obius group actions in the solvable chimera model}

\author{Vladimir Ja\' cimovi\' c \\
	Faculty of Natural Sciences and Mathematics\\
	University of Montenegro\\ 
	Cetinjski put bb., 81000 Podgorica\\ 
	Montenegro\\
	\texttt{vladimirj@ucg.ac.me} \\
	\And
	Aladin Crnki\' c \\
	Faculty of Technical Engineering\\
	University of Biha\' c\\
	Irfana Ljubijanki\' ca bb., 77000 Biha\' c\\
	Bosnia and Herzegovina\\
	\texttt{aladin.crnkic@unbi.ba} \\
}

\renewcommand{\shorttitle}{M\" obius group actions in the solvable chimera model}

\hypersetup{
pdftitle={M\" obius group actions in the solvable chimera model},
pdfsubject={-},
pdfauthor={Vladimir Ja\' cimovi\' c, Aladin Crnki\' c},
pdfkeywords={Kuramoto model, chimera state, M\" obius transformation, Poisson kernel},
}

\begin{document}
\maketitle

\begin{abstract}
We study actions of M\" obius group on two sub-populations in the solvable chimera model proposed by Abrams et al. Dynamics of global variables are given by two coupled Watanabe-Strogatz systems, one for each sub-population.

At the first glance, asymptotic dynamics in the model seem to be very simple. For instance, in the stable chimera state distributions of oscillators perform a simple rotations after a certain (sufficiently large) moment. However, a closer look unveils that dynamics are subtler that what can be observed from evolution of densities of oscillators' phases. In order to gain the full picture, one needs to investigate dynamics on the transformation group that acts on these densities. Such an approach emphasizes impact of the "hidden" variable that is not visible on macroscopic level.
\end{abstract}


\section{Introduction}\label{sec:1}
Populations of coupled oscillators display complicated spatiotemporal patterns even in the simple setup with identical, symmetrically coupled oscillators. For instance, such a population can spontaneously split into synchronized and desynchronized sub-populations. The first model that exhibits such kind of a state has been reported in 2002. by Kuramoto and Battogtokh \cite{KB}. Existence of such states in homogeneous populations came as a surprise and posed certain mathematical challenges for the research community in the field \cite{Omelchenko,PA}. This puzzling phenomenon has later been named {\it chimera state} or, simply, {\it chimera} \cite{AS}.
Following the pioneering model of Kuramoto and Battogtokh, existence of chimeras has been numerically confirmed in several other models \cite{BPR,Laing,MLS}.

It took nearly a decade before chimera states have been observed in some real-life experiments, including reactions with chemical oscillators \cite{TNS} and mechanical devices (metronomes) \cite{MTFH}. More recently, chimera states have also been reported in arrays of nano-oscillators (spintronics) \cite{ZP} and in social-type networks \cite{Pikovsky}.

The simplest, analytically tractable model that admits chimera state has been introduced in 2008. by Abrams et al. \cite{AMSW}. This model describes two sub-populations in which each oscillator is coupled to all the others, but the coupling within each sub-population is stronger than the coupling between different sub-populations. In addition, the coupling function includes a phase shift. If we denote the two sub-populations by $A$ and $B$ the governing equations in the model are \cite{AMSW}
\begin{equation}
\label{chimera}
\frac{d \varphi_j^l}{dt} = \omega + \sum \limits_{k=A,B} \frac{K_{kl}}{N_k} \sum \limits_{i=1}^{N_k} \sin(\varphi_i^k - \varphi_j^l - \beta), \quad l=A,B, \quad j=\overline{1,N_l}.
\end{equation}
Here, $N_l, \, l=A,B$ are numbers of oscillators in sub-populations $A$ and $B$. Oscillators have equal intrinsic frequencies $\omega$ and there is a global (i.e. common for each pair of oscillators) phase shift $\beta$ in the coupling function. The coupling strength between oscillators that belong to the same sub-population equals $K_{AA} = K_{BB} = \eta > 0$, while the coupling between different sub-populations is $K_{AB} = K_{BA} = \nu > 0$. It is assumed that the coupling within each sub-population is stronger i.e. $\eta > \nu$.

Numerical simulations suggest that there exists a set of values of parameters $\eta, \nu$ and $\beta$ for which (\ref{chimera}) is multistable. In some simulations the system achieves completely synchronous state, with all phases $\varphi_i^A, \varphi_j^B$ equal. In the second scenario one sub-population achieves coherence, while the other remains only partially synchronized. This second scenario is particularly interesting, because it happens for symmetric initial conditions and completely symmetric coupling for the two sub-populations. Hence, it is natural to wonder which factors determine if the system will synchronize, and (if the chimera scenario is realized) which sub-population will achieve coherence. Furthermore, it has been observed that, depending on the parameter values, desynchronized sub-population can exhibit different asymptotic behaviors. For some values of $\eta, \nu$ and $\beta$, the system admits so-called {\it stable chimera state}, in which the order parameter of desynchronized sub-population achieves a certain equilibrium value between zero and one. If the difference $A = \eta - \nu$ is increased above a certain threshold, the chimera state loses stability and the order parameter of desynchronized sub-population starts to oscillate. This corresponds to less trivial collective dynamics within desynchronized sub-population. This second type of asymptotic behavior has been named {\it breathing chimera} in \cite{AMSW}. If we further increase the difference $A = \eta - \nu$, the multistability is lost, chimera state disappears, and the system always converges towards completely synchronous state.

Most important, Abrams et al. have shown that numerical simulations are not necessary for investigation of asymptotic behaviors in (\ref{chimera}). The model is {\it solvable} under suitable assumptions. If numbers of oscillators $N_A$ and $N_B$ are large and initial distributions of oscillators in both sub-populations are uniform, the system can be studied analytically. Order parameters $r_A$ and $r_B$ of the two sub-populations satisfy a certain system of ODE's. Then the bifurcation analysis of this system of ODE's unveils existence of three stable equilibrium states that correspond to the complete synchronization or to chimeras.

The underlying fact behind this analytic approach is that, due to the result of Marvel et al. \cite{MMS}, oscillators in both sub-populations evolve by actions of M\" obius group. In mathematical terms, (\ref{chimera}) induces a one-parametric family of biholomorphic mappings of the bi-disc. In other words, the system evolves by two mutually coupled actions of the M\" obius group on two unit discs.

The idea of studying (\ref{chimera}) through actions of biholomorphic mappings is not new. For instance, Pikovsky and Rosenblum have studied multi-population systems based on the Watanabe-Strogatz reductions in \cite{PR}. In \cite{EM1} Engelbrecht and Mirollo have used the M\" obius group approach to investigate possible equilibrium states for the Kuramoto model with identical oscillators and global coupling. In the concluding remarks of their paper authors suggest that the same approach is also suitable for study of system (\ref{chimera}).

Our approach in the present paper is based on this observation. We will be concerned with the question of what is happening once the chimera state is achieved. Since asymptotic behavior of the synchronized sub-population is clear, we will focus on desynchronized sub-population. Our analysis unveils that asymptotic dynamics in the chimera state are deceptively simple and that some unexpected and, in the certain sense, invisible effects are present.

The main intrigue of the present paper will be explained more precisely at the end of the next Section. Before that we briefly recall some previous results that constitute the basis of our investigation.

\section{Low dimensional dynamics on orbits of the M\" obius group}\label{sec:2}
In their seminal paper \cite{WS}, Watanabe and Strogatz have shown that the Kuramoto model with identical, globally coupled oscillators admits many constants of motion. The dynamics in this model can be reduced to the system of only three ODE's regardless of the total number of oscillators. It took almost 15 years until underlying symmetries have been explained from geometric and group-theoretic point of view in \cite{MMS}. Marvel et al. have shown that identical, globally coupled oscillators evolve by actions of the M\" obius group. The 3-dimensional system of ODE's, derived by Watanabe and Strogatz, defines dynamics on the M\" obius group (i.e. on the group of conformal mappings that preserve the unit disc).

In order to explain this, start with the group of M\" obius transformations in the complex plane. Denote by $G$ the set of all transformations that leave the unit disc invariant. $G$ is subgroup of the larger group of all M\" obius transformations of the complex plane. This (sub)group is referred to as {\it M\" obius group} in the title and throughout the present paper. The general transformation from $G$ can be written in the following form
\begin{equation}
\label{Mobius}
g(z) = \frac{e^{i \psi} z + \alpha}{1 + \bar \alpha e^{i \psi} z}.
\end{equation}
Parameters of this transformation are the angle $\psi \in [0,2 \pi]$ and the complex number $\alpha \in {\mathbb C}$, $|\alpha| < 1$. Hence, (real) dimension of group $G$ equals three.

By taking substitution $z_j^l = e^{i \varphi_j^l}$ for $j=1,\dots,N_l$, $l=A,B$, we represent phase oscillators by points on the unit circle. All oscillators in (\ref{chimera}) are identical and oscillators belonging to the same sub-population are coupled to all other oscillators in the identical way. In other words, each sub-population is subject to its own mean field and these two mean fields induce two actions of $G$. Then the simple adaption of the result of Marvel et al. \cite{MMS} yields the following

\begin{proposition}
\label{prop:1}
Consider a population of oscillators governed by eqs. (\ref{chimera}).
There exist two one-parametric families $g_t^A$ and $g_t^B$ of M\" obius transformations from $G$, such that
$$
z_j^A(t) = g_t^A(z_j^A(0)), \quad j=1,\dots,N_A; \\
$$
$$
z_j^B(t) = g_t^B(z_j^B(0)), \quad j=1,\dots,N_B.
$$
\end{proposition}

Furthermore, parameters of the families $g_t^A$ and $g_t^B$ satisfy the following system of ODE's
\begin{equation}
\label{WS^2}
\left \{
\begin{array}{llll}
\dot \alpha_A = i (f_A \alpha_A^2 + \omega \alpha_A + \bar f_A); \\
\dot \psi_A = f_A \alpha_A + \omega + \bar f_A \bar \alpha_A; \\
\dot \alpha_B = i (f_B \alpha_B^2 + \omega \alpha_B + \bar f_B); \\
\dot \psi_B = f_B \alpha_B + \omega + \bar f_B \bar \alpha_B
\end{array}
\right.
\end{equation}
with coupling functions
$$
f_A = \frac{i \eta}{2 N_A} \sum \limits_{j=1}^{N_A} e^{-i(\varphi_j^A - \beta)} + \frac{i \nu}{2 N_B} \sum \limits_{j=1}^{N_B} e^{-i(\varphi_j^B - \beta)},
$$
$$
f_B = \frac{i \eta}{2 N_B} \sum \limits_{j=1}^{N_B} e^{-i(\varphi_j^B - \beta)} + \frac{i \nu}{2 N_A} \sum \limits_{j=1}^{N_A} e^{-i(\varphi_j^A - \beta)}.
$$
If we denote by $c_A$ and $c_B$ centroids of sub-populations $A$ and $B$ respectively and assume that $N_A = N_B$, the expressions for coupling functions can be simplified as
$$
f_A = \frac{i}{2} e^{i \beta} (\eta \bar c_A + \nu \bar c_B) \mbox{  and  } f_B = \frac{i}{2} e^{i \beta} (\eta \bar c_B + \nu \bar c_A),
$$
where the notion $\bar w$ stands for the complex conjugation of a complex number $w$.

Underline that (\ref{WS^2}) is the dynamical system on the manifold (Lie group) $G \times G$ of (real) dimension 6, where equations are mutually coupled through complex-valued functions $f_A$ and $f_B$. This means that (\ref{chimera}) generates a trajectory on $G \times G$. This trajectory is determined by initial positions $z_j^A(0)$ and $z_j^B(0)$ of oscillators. In its turn, this evolution on $G \times G$ determines one-parametric families of transformations $g^A_t$, $g_t^B$ whose actions govern evolution of oscillators.

Furthermore, the system evolves on orbits of group $G \times G$ and, since this group is 6-dimensional, the evolution takes place on a 6-dimensional invariant submanifolds (product of two 3-dimensional invariant submanifolds, one for each sub-population). This 6-dimensional submanifold is determined by the initial state of the system.

\subsection{Evolution on the Poisson manifold}
The most interesting is the special case when, due to rotational symmetry of the initial state, the system evolves on 4-dimensional invariant submanifold (that is - 2-dimensional submanifold for each sub-population). To that end we introduce the following

\begin{assumption}
\label{assum:1}
Consider the system (\ref{chimera}) in thermodynamic limit $N_A, N_B \to \infty$ and assume that the initial distribution of oscillators in both sub-populations is uniform on the unit circle $S^1$.
\end{assumption}

In other word, we assume that the initial distribution of oscillators in both sub-populations is given by a density function $\rho(\varphi) = \frac{1}{2 \pi}, \varphi \in [0,2 \pi]$.

\begin{remark}
Mathematically rigorous way to describe this thermodynamic limit would be to pass from (\ref{chimera}) to the system of two first order integro-partial differential equations that govern evolution of the oscillator's densities on the unit circle. This has been done in previous papers, for instance \cite{MS}. The relation between N-finite Kuramoto model and its thermodynamic limit poses some difficulties and subtle questions that have been discussed in \cite{Strogatz}. We will go around these difficulties, as it would require longer exposition.
\end{remark}

\begin{proposition}
\label{prop:2}
Under Assumption \ref{assum:1} distributions of oscillators in sub-populations $A$ and $B$ at each moment $t$ are given by the following density functions
\begin{equation}
\label{Poisson}
\rho_l(t,\varphi) = \frac{1}{2 \pi} \frac{1-r_l^2(t)}{1 - 2 r_l(t) \cos(\varphi - \psi_l(t)) + r_l^2(t)} \mbox{  for } l=A,B.
\end{equation}
Here $0 \leq r_l < 1$ and $0 \leq \psi_l < 2 \pi$.
\end{proposition}

\begin{remark}
Functions of the form (\ref{Poisson}) are well known in Mathematics as {\it Poisson kernels} on the unit circle. They are density functions for harmonic measures on $S^1$. They constitute 2-dimensional invariant submanifold for evolution (\ref{chimera}) which is called {\it Poisson manifold}. Indeed, each Poisson kernel (\ref{Poisson}) is uniquely determined by a point $\alpha_l = r_l e^{\Phi_l}$ in the unit disc $\mathbb{D}$. In fact, there is one-to-one (conformally natural) map from the Poisson manifold to the hyperbolic unit disc $\mathbb{D}$. In the context of coupled oscillators, Poisson manifold is also called {\it Ott-Antonsen manifold}. However, this second term is typically used for the (slightly more complicated) manifold that appears after the Ott-Antonsen reduction of the Kuramoto model with non-identical oscillators, see \cite{OA}.
\end{remark}

\begin{remark}
Uniform measure and the delta distribution are also Poisson kernels, obtained from (\ref{Poisson}) for $r_l = 0$ and the limit case $r_l \to 1$.
\end{remark}

\begin{proposition}
\label{prop:3}
Complex number $\alpha_l$ is the mean value (complex order parameter) for the density (\ref{Poisson}). Then, $r_l = |\alpha_l|$ is the real order parameter for (\ref{Poisson}).
\end{proposition}

Using Proposition \ref{prop:3} one can simplify (\ref{WS^2}) to obtain the system for real order parameters $r_A$ and $r_B$. Indeed, in this case ODE's for $\psi_A$ and $\psi_B$ in (\ref{WS^2}) decouple from those for $\alpha_A$ and $\alpha_B$. Hence, ODE's for $\psi_A$ and $\psi_B$ can be neglected. Then, ODE's for $\alpha_A$ and $\alpha_B$ yield four real-valued ODE's. Using rotational invariance, one can introduce the phase shift $\Phi_d = \Phi_B - \Phi_A$ and, without loss of generality, assume that sub-population $A$ is synchronized, i.e. $r_A = 1$. Then we are left with the system for variables $r_B$ and $\psi_d$
\begin{figure}[t]
\centering
  \begin{tabular}{@{}cc@{}}
    \includegraphics[width=.38\textwidth]{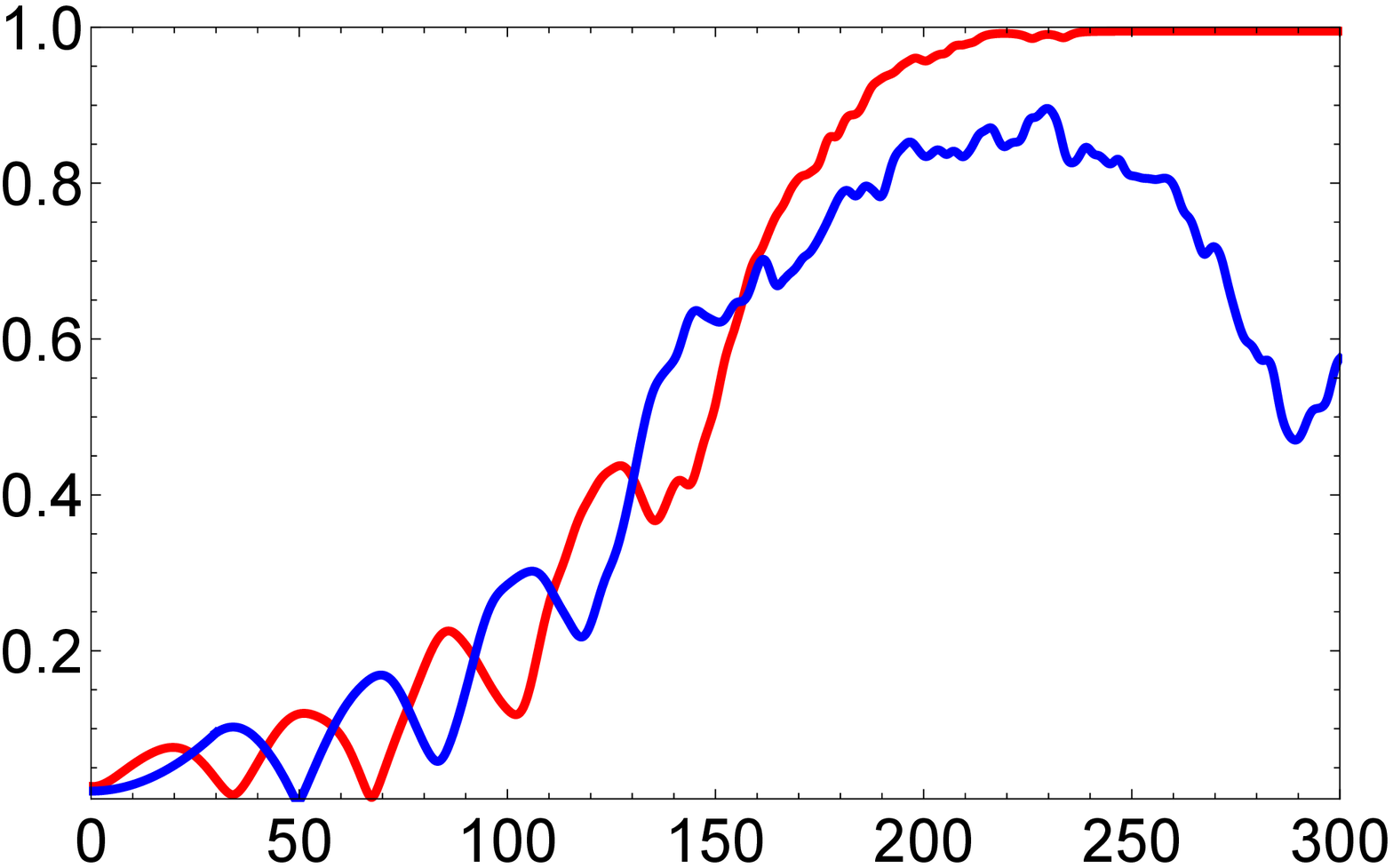} &
    \includegraphics[width=.383\textwidth]{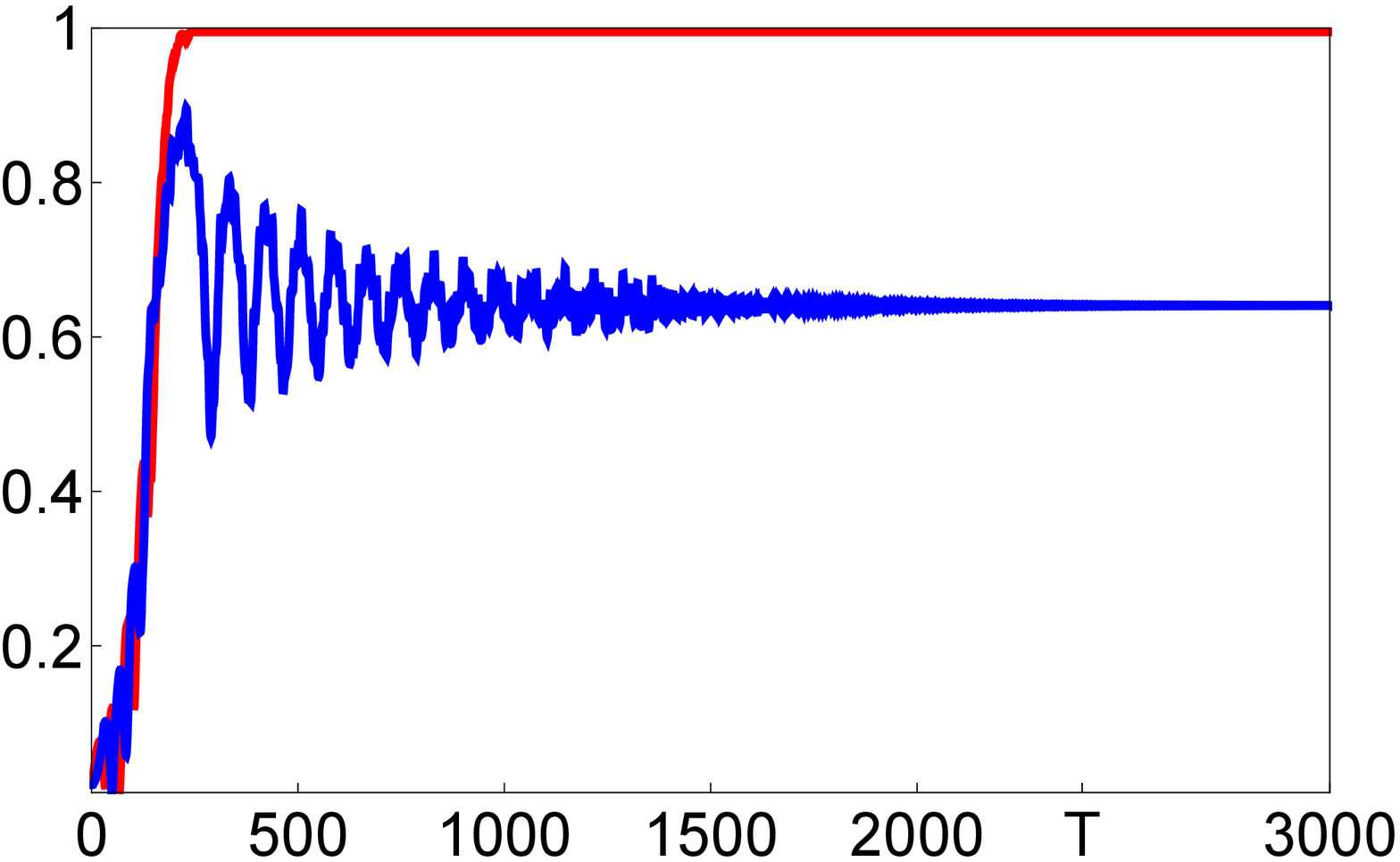}\\
    \quad (a)&\quad (b)
  \end{tabular}
  \caption{\label{fig:1}Real order parameters $r_A(t)$ and $r_B(t)$ for sub-populations $A$ and $B$ in model (\ref{chimera}) with $\mu = 0.623, \nu = 0.377, \beta = \frac{\pi}{2} - 0.1$ (stable chimera) and uniform initial distributions on time intervals (a) $t \in [0,300]$ and (b) $t \in [0,3000]$.}
\end{figure}

\begin{equation}
\label{r_B}
\left \{
\begin{array}{ll}
\dot r_B = \frac{1 - r_B^2}{2}[\eta r_B \cos \beta + \nu \cos(\psi_d - \beta)];\\
\dot \psi_d = \frac{1+r_B^2}{2 r_B} [\eta r_B \sin \beta - \nu \sin(\psi_d - \beta)] - \eta \sin \beta - \nu r_B \sin(\psi_d + \beta).
\end{array}
\right.
\end{equation}

The derivation of (\ref{r_B}) and the consequent bifurcation analysis has been reported in \cite{AMSW} (see also \cite{PA}). This system is simple enough that it can be studied analytically. First of all, it is easy to check that (\ref{r_B}) always has a stable equilibrium with $r_B = 1$. This equilibrium corresponds to the full synchronization. However, simple analysis shows that for certain parameter values, there also exists another stable equilibrium with $0 < r_B < 1$. This second equilibrium corresponds to the stable chimera state. \footnote{In other words, (\ref{r_B}) is bi-stable. However, the full system with the additional ODE for $r_A$ is three-stable: it admits full synchronization ($r_A = r_B = 1$) and two symmetric chimeras ($r_A < 1, r_B = 1$ and vice versa).}

Furthermore, if parameter $A = \eta - \nu$ goes through a certain threshold this equilibrium loses stability and $r_B$ starts to oscillate. This can be checked analytically by using the Hopf bifurcation theorem. This switch to periodic oscillations of $r_B$ leads to a new phenomenon, that is named {\it breathing chimera} in \cite{AMSW}.

We are now in position to explain the main point of the present paper. For instance, adopt Assumption \ref{assum:1} and consider the stable chimera. Focus on desynchronized sub-population $B$ and fix the moment $T$ sufficiently large, so that $r_B(t) = const$ for all $t>T$ (see Figure \ref{fig:1}b). By summarizing everything said in this Section we can state the two facts.

\begin{fact}
\label{fac:1}
From \cite{MMS} we know that all oscillators evolve by actions of the M\" obius group. In other words, for each $t>T$ there exists a transformation $g \in G$, such that $z^B_j(t) = g(z^B_j(T))$, for $j=1,\dots,N_B$.
\end{fact}

On the other hand, Proposition \ref{prop:2} asserts that the density of oscillators in desynchronized group is given by (\ref{Poisson}). But, densities (\ref{Poisson}) are uniquely defined by their mean values $\alpha_B$. Since $r_B(t) = |\alpha(t)| = const$ for $t > T$, this means that $\alpha_B(t)$ performs simple rotations in the unit disc after moment $T$. This brings as to the following

\begin{fact}
\label{fac:2}
Density of oscillators in sub-population $B$ performs simple rotations after moment $T$. In other words, for each $t > T$ the density $\rho_B(t,\varphi)$ is a simple rotation of $\rho_B(T,\varphi)$.
\end{fact}

This brings us to the main question of the present paper.

\begin{question}
\label{que:1}
Can we conclude from facts \ref{fac:1} and \ref{fac:2} that all oscillators in sub-population $B$ evolve by simple rotations after moment $T$?
\end{question}

This question is the starting point of the present investigation. In order to put it as simple as possible, we pass to the rotating coordinate frame. Then the distribution of oscillators in sub-population $B$ is stationary after moment $T$ and Question \ref{que:1} can be reformulated as follows: Is it possible that all oscillators evolve by the same one-parametric family of non-trivial M\" obius transformations, while their density remains stationary?

There are two hints indicating that there might be something subtler in asymptotic dynamics of the desynchronized sub-population. First, it has been reported previously that oscillators in desynchronized sub-population satisfy a pretty complicated equation (see eq. (6) in \cite{Laing}). Second, recent findings of \cite{EM2} suggest that reduction of the classical Kuramoto model (with a single population) to the Poisson manifold are not as simple as it appears at the first glance.

In sections \ref{sec:4} and \ref{sec:5} we will demonstrate that asymptotic dynamics in desynchronized sub-population are indeed more involved then they appear. In Section \ref{sec:4} we will conduct simulations in order to analyze both the stable and the breathing chimera. In Section \ref{sec:6} we will abandon Assumption \ref{assum:1} in order to investigate what is happening on generic 3-dimensional invariant sub-manifolds. Our findings will underline some crucial qualitative differences between dynamics on and off the Poisson manifold.

Before proceeding further, we introduce more rigorous mathematical framework and clarify some notations. This will be done in the next Section.

\section{Mathematical setup and preliminaries}\label{sec:3}
Group $G$ operates on the unit circle $S^1$, on unit disc $\mathbb{D}$ and on space ${\cal P}(S^1)$ of probability measures on the circle. Explicitly
$$
g \cdot z = g(z), \mbox{  if  } z \in S^1 \cup \mathbb{D};
$$
$$
g \cdot \mu (A) = g_* \mu(A) = \mu(g^{-1}(A)), \mbox{  if  } \mu \in {\cal P}(S^1) \mbox{  and  } A \subseteq S^1 \mbox{  is a Borel set}.
$$
Hence, $g_* \mu$ denotes a measure obtained by the action of $g \in G$ on measure $\mu$.

The state of sub-population $B$ at moment $t$ will be represented by a probability measure $\mu_B(t) \in {\cal P}(S^1)$. If $\mu_B(0)$ is the initial state, then the state at each moment $t$ is
$$
\mu_B(t) = g^B_{t \; \; *} \mu_B(0).
$$

It is obvious from (\ref{Mobius}) that parameters $\alpha_A$ and $\alpha_B$ are images of zero under maps $g^A_t$ and $g^B_t$ respectively, that is
$$
\alpha_A(t) = g^A_t(0) \mbox{  and  } \alpha_B(t) = g_t^B(0).
$$

As already explained, system (\ref{chimera}) generates a trajectory on group $G \times G$. In the next Section we will investigate this trajectory through simulations. This will be done by observing action of families $g_t^A$ and $g_t^B$ on zero (center of the disc). Hence, we will simulate system (\ref{chimera}) and depict $\alpha_A(t)$ and $\alpha_B(t)$ in an attempt to understand the corresponding dynamics on $G \times G$. We will be mainly interested in family $g_t^B$ that acts on the desynchronized sub-population.

Notice, however, that transformation $g \in G$ is not uniquely determined by its action on zero (i.e. by the point $\alpha$). Indeed, there are infinitely many M\" obius transformations that map zero to a given point $\alpha \in \mathbb{D}$. In essence, we will depict the projection of trajectory on $G \times G$ on $\mathbb{D} \times \mathbb{D}$. This projection map has fibers $S^1 \times S^1$ on which variables $\psi_A$ and $\psi_B$ "live".

In whole, simulations will not give us the full picture of the trajectory on $G \times G$. However, for our line of reasoning the following simple fact is important

\begin{proposition}
\label{prop:4}
$g(0) = 0$ if and only if $g$ is a simple rotation.
\end{proposition}

Hence, the transformation $g_t^B$ is a rotation if and only if $\alpha_B(t) = g_t^B(0) = 0$. Moreover, $g_t^B$ is identity map for $t=0$.

Furthermore, we will be interested in action of families $g_t^A$ and $g_t^B$ on different time intervals. This requires a refinement of notations.

Suppose that the initial state is $\mu(0)$ and the state at the moment $t$ is $\mu(t)$. Then, there exists a M\" obius transformation $m \in G$, such that $\mu(t) = m_* \mu(0)$. We will denote this transformation $m$ by $g_t = g_{[0,t]}$. Hence,
$$
\mu(t) = g_{[0,t] \; *} \mu(0).
$$
In an analogous way, if states at moments $t_1$ and $t_2$ are given by $\mu(t_1)$ and $\mu(t_2)$ respectively, then there exists $m \in G$, such that $\mu(t_2) = m_* \mu(t_1)$. We will denote this transformation $m$ by $g_{[t_1,t_2]}$. Hence,
$$
\mu(t_2) = g_{[t_1,t_2] \; *} \mu(t_1).
$$
Furthermore, we will use notation $\alpha(t_1;t_2)$ for the image of zero under M\" obius transformation $g_{[t_1,t_2]}$, that is
$$
\alpha(t_1;t_2) = g_{[t_1,t_2]}(0).
$$
Underline that $\alpha(t;t) = 0$ for any $t$, since $g_{[t,t]}$ is the identity map.

In order to understand asymptotic behavior of the trajectory on group $G \times G$ we will depict paths $\alpha(0;t)$ and $\alpha(T;t)$ in the unit disc. This will enable us to answer Question \ref{que:1}.
\begin{figure}[t]
\centering
  \begin{tabular}{@{}cc@{}}
    \includegraphics[width=.335\textwidth]{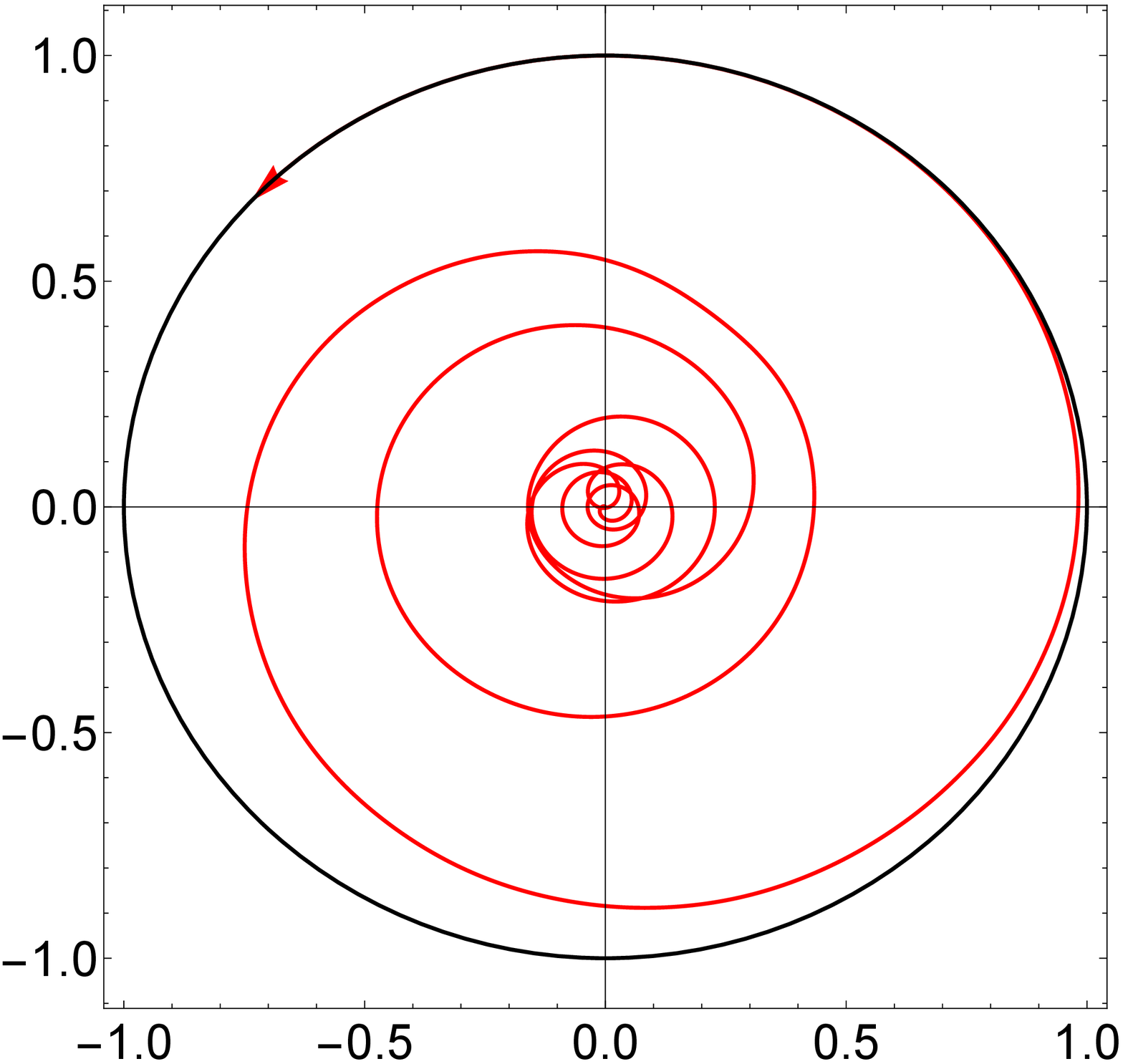}&
    \includegraphics[width=.335\textwidth]{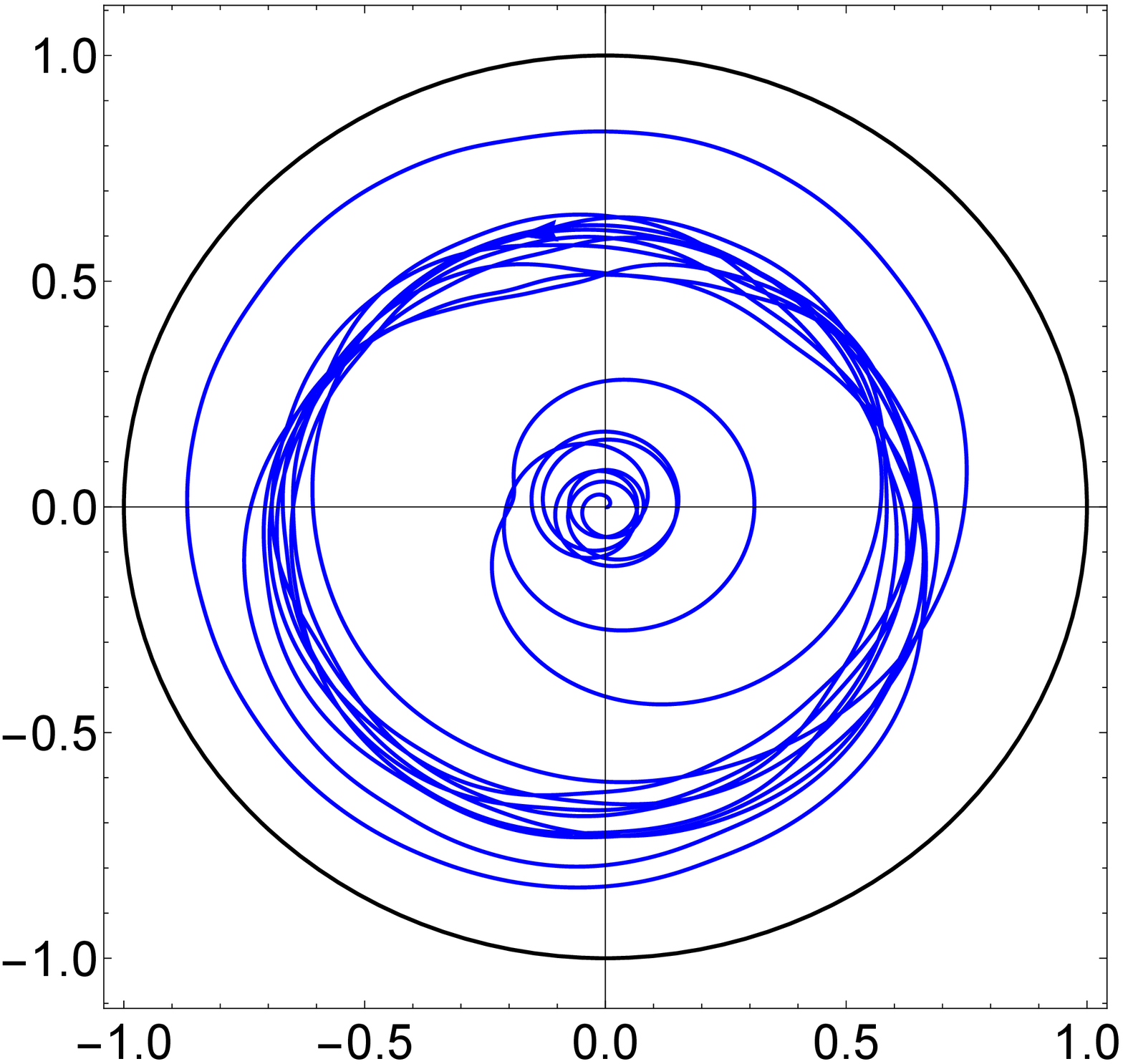}\\
 \quad (a)&\quad (b)
  \end{tabular}
  \caption{\label{fig:2}M\" obius group actions as the system (\ref{chimera}) evolves towards a stable chimera state: (a) $\alpha_A(0;t)$ and (b) $\alpha_B(0;t)$ on time interval $t \in [0,1000]$. Initial distributions of oscillators are uniform.}
\end{figure}

\section{Simulations: M\" obius transformations in stable and breathing chimeras on the Poisson manifold}\label{sec:4}
Throughout this Section we work under Assumption \ref{assum:1}, i.e. we sample initial positions for both sub-populations from the uniform distribution on the circle. Then, due to Proposition \ref{prop:2}, states of both sub-populations evolve on the Poisson manifold. We pick $N_A = N_B = 500$ oscillators in each simulation.

After sampling initial positions, we solve (\ref{chimera}) and choose three arbitrary oscillators from each sub-population. We compute positions of these oscillators at each moment $t$. Positions of these triples of points at moments $t_1$ and $t > t_1$ uniquely determine M\" obius transformations $g_{[t_1,t]}^l, \, l=A,B$. We depict actions of these transformations on zero, i.e. the point $\alpha_l(t_1;t)$.

Since $g_{[t,t]}^l$ is the identity transformation, $\alpha_l(t_1;t)$ starts from zero in all simulations.

\subsection{Case I: Simulations of stable chimera at the Poisson manifold}

Figure \ref{fig:1} shows evolution of real order parameters of the two sub-populations as the system evolves towards the chimera state. We observe an unpredictable competition where sub-populations alternately increase their coherence at the expense of each other. At a certain moment one sub-population rapidly synchronize leaving another desynchronized.

Figure \ref{fig:2} demonstrates actions of the two families of M\" obius transformations at time intervals $[0,t]$. For the synchronized group $\alpha_A(0;t) = g_{[0,t]}^A(0)$ goes to infinitely distant horizon, i.e. $|\alpha_A(0;t)| \to 1$ when $t \to \infty$. On the other hand, $\alpha_B(0;t)$ is constrained inside a circle in the unit disc, i.e. $|\alpha_B(0;t)| \to r_B(\infty) < 1$ for all $t$.

Simulation results presented in Figure \ref{fig:2} are expected and do not help to answer Question \ref{que:1}. In order to access asymptotic behavior of the desynchronized sub-population, fix a sufficiently large moment $T$, such that $r_B(t) = const$ for $t>T$. In this way we investigate what happens once the stable chimera is achieved. We investigate transformations acting on state $\mu_B(T)$. To this end we depict the path $\alpha_B(T;t) = g^B_{[T,t]}(0)$. As shown in Figure \ref{fig:3}, $\alpha_B(T,t)$ does not stay at zero at all times, however it repeatedly turns to zero.
\begin{figure}[t]
\centering
  \begin{tabular}{@{}c@{}}
    \includegraphics[width=.335\textwidth]{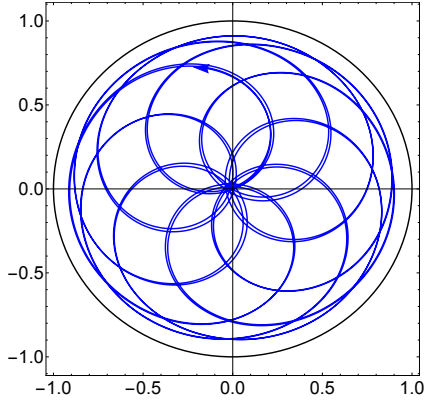}
  \end{tabular}
  \caption{\label{fig:3}M\" obius group action for desynchronized sub-population $B$ in stable chimera state: $\alpha_B(1000;t)$ at time interval $t \in [1000,1500]$. Notice that $\alpha_B(1000;t)$ returns to zero infinitely many times.}
\end{figure}

These results demonstrate that oscillators $z_j^B(t)$ in desynchronized sub-population {\bf do not} evolve by simple rotations at the state of stable chimera. \footnote {Indeed, due to Proposition \ref{prop:4}, if they would evolve by simple rotations, then $\alpha_B(T,t)$ would be zero for all moments $t>T$.}

However, there exists a sequence of moments $t_1,t_2,\ldots$, such that $\alpha_B(T;t_1) = \alpha_B(T;t_2) = \cdots = 0$. This implies that $g^B_{[T,t_1]}, g^B_{[T,t_2]}, \ldots$ {\bf are} simple rotations.

In whole, dynamics in the state of stable chimera is much more involved then it looks at the first glance. In fact, it is a bit puzzling: each individual oscillator evolves by non-trivial M\" obius transformations, while their density evolves by simple rotations. This will be discussed in sections \ref{sec:5} and \ref{sec:7}. Before that, we examine what is happening in the state of the breathing chimera.

\subsection{Case II: Simulations of breathing chimera at the Poisson manifold}

By slightly modifying coupling strengths in system (\ref{chimera}) we obtain the breathing chimera (compare parameter values in captions under figures \ref{fig:1} and \ref{fig:4}). Then, equilibrium $r_B(t) \equiv a$ in system (\ref{r_B}) loses stability, undergoes the Hopf bifurcation and $r_B(t)$ starts to oscillate, as shown in \cite{AMSW}. Figure \ref{fig:4} shows evolution of real order parameters as the system evolves towards breathing chimera.

Figure \ref{fig:5} provides an insight into actions of $g_{[0,t]}^A$ and $g_{[0,t]}^B$ for synhronized and desynchronized sub-population respectively. We do not see qualitative differences with the stable chimera, Figure \ref{fig:2}.

Unlike the stabe chimera, real order parameter $r_B(t)$ is not constant after certain moment $T$. This means that densities of desynchronized sub-population at different times are not related by a simple rotations. However, there exists a sequence $t_1, t_2, \ldots$, such that $r_B(T) = r_B(t_1) = r_B(t_2) = \cdots$. Hence, densities at $T, t_1, t_2, \ldots$ {\bf are} related by simple rotations.
\begin{figure*}[t]
\centering
  \begin{tabular}{@{}cc@{}}
    \includegraphics[width=.38\textwidth]{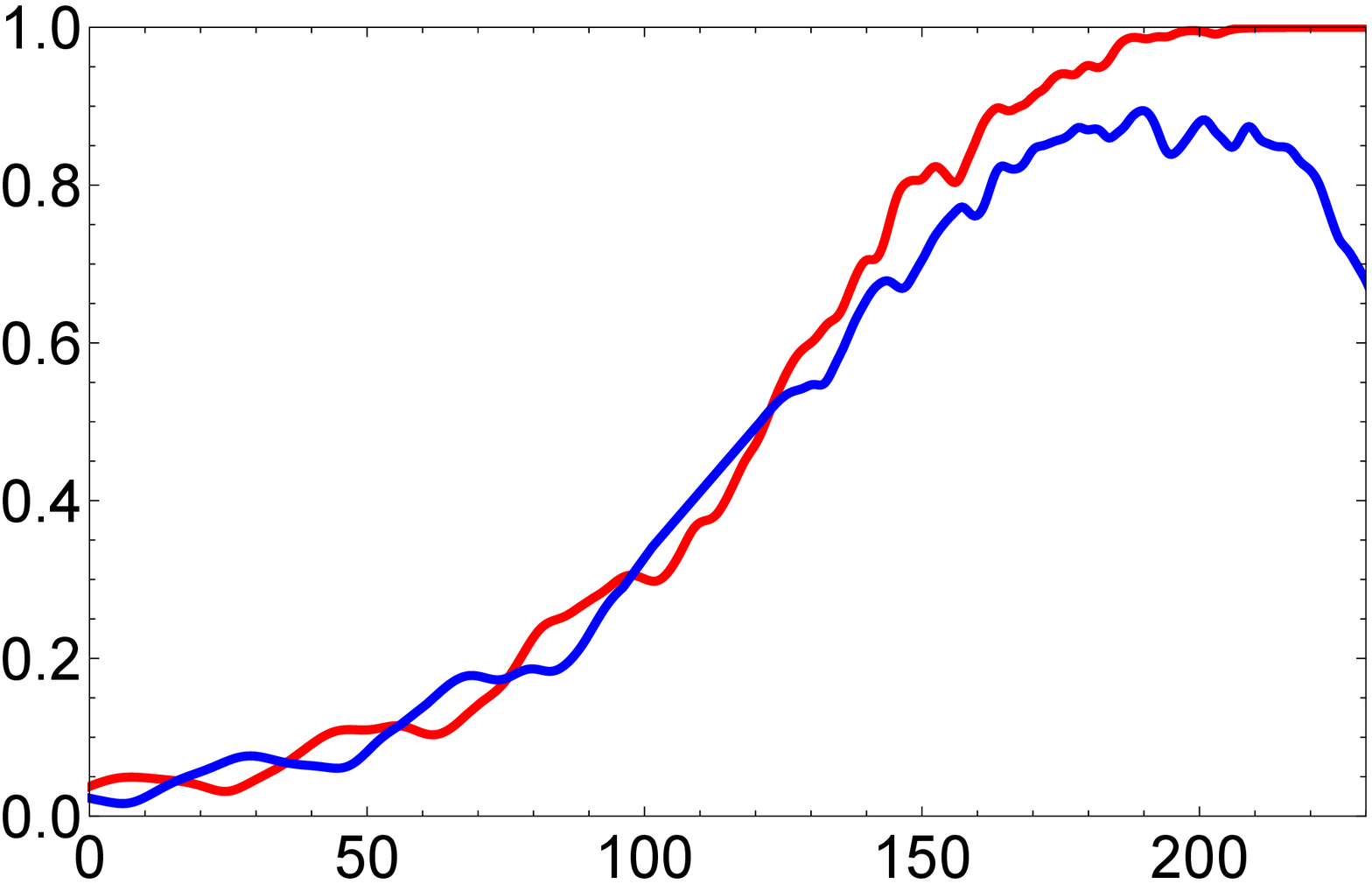}&
     \includegraphics[width=.397\textwidth]{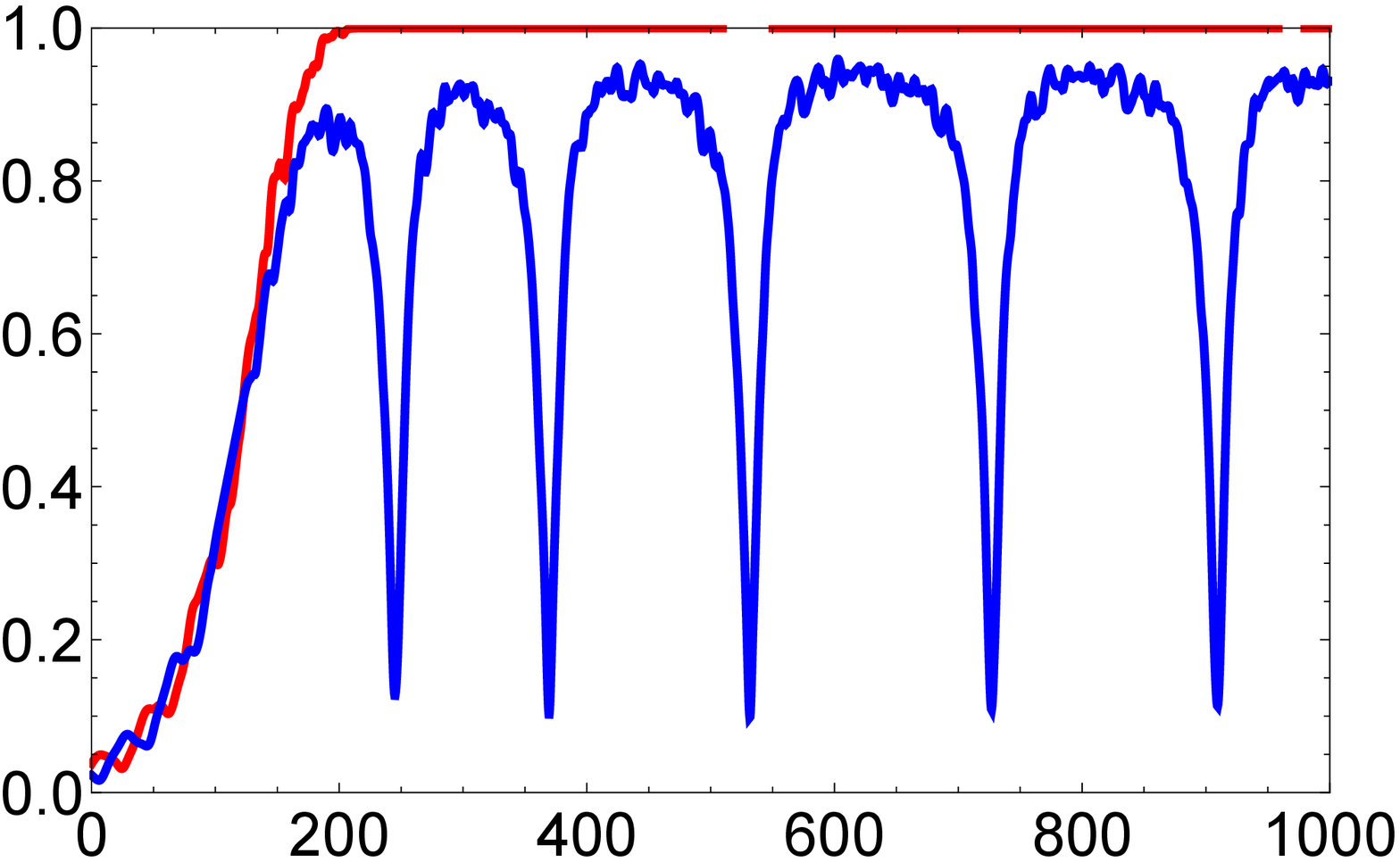}\\
\quad (a)&\quad (b)
  \end{tabular}
  \caption{\label{fig:4}Real order parameters $r_A(t)$ and $r_B(t)$ for sub-populations $A$ and $B$ in model (\ref{chimera}) with $\mu = 0.675, \nu = 0.325, \beta = \frac{\pi}{2} - 0.1$ (breathing chimera) and uniform initial distributions on time intervals (a) $t \in [0,230]$ and (b) $t \in [0,1000]$.}
\end{figure*}
\begin{figure*}[t]
\centering
  \begin{tabular}{@{}cc@{}}
    \includegraphics[width=.33\textwidth]{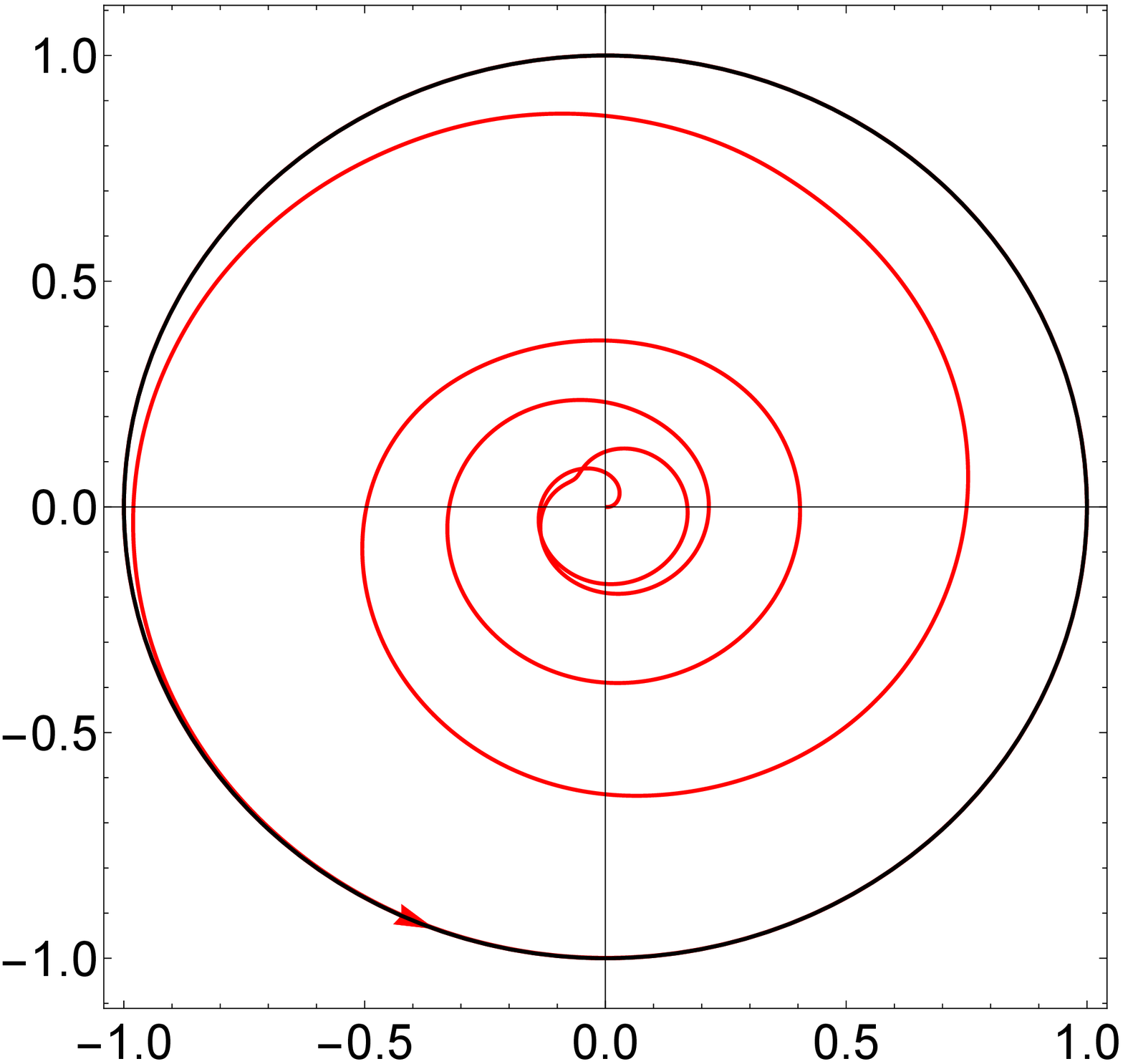}&
    \includegraphics[width=.33\textwidth]{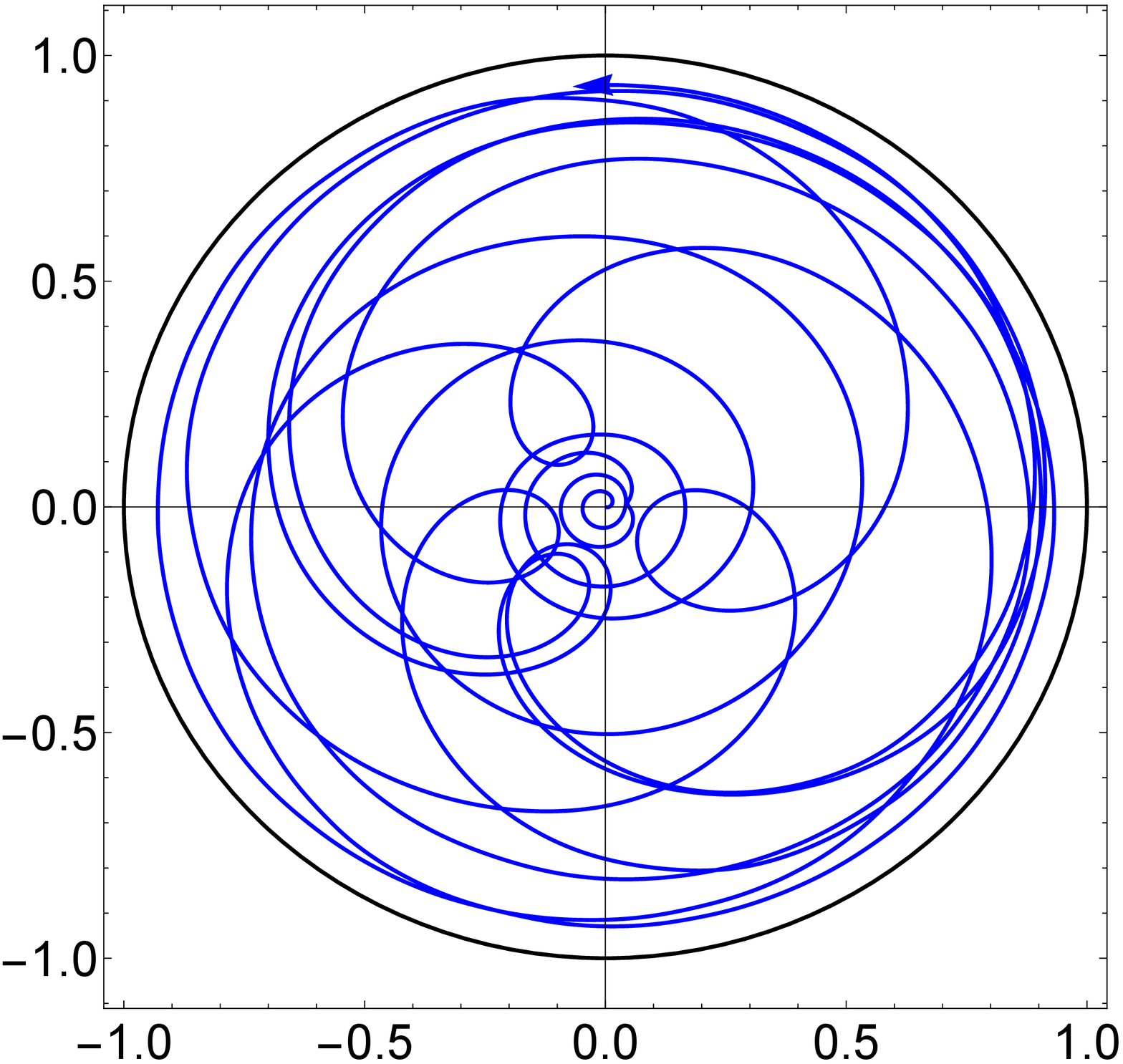}\\
\quad (a)&\quad (b)
  \end{tabular}
  \caption{\label{fig:5}M\" obius group actions as the system (\ref{chimera}) evolves towards breathing chimera: (a) $\alpha_A(0;t)$ and (b) $\alpha_B(0;t)$ on time interval $t \in [0,1000]$. Initial distributions of oscillators are uniform.}
\end{figure*}

This brings us to the following question: Are transformations $g_{[T,t_1]}^B, g_{[T,t_2]}^B, \cdots$ simple rotations?

In order to answer this question, we depict actions of $g_t^A$ and $g_t^B$ after certain (sufficiently large) moment $T$. Results are shown in Figure \ref{fig:6}. We observe one important difference with the stable chimera (Figure \ref{fig:3}): now parameter $\alpha_B(T;t)$ never returns to zero. This suggests that $g_{[T,t]}^B$ {\bf are never} rotations, not even at moments $t=t_1, t= t_2, \cdots$.

\section{Discussion: M\" obius group actions in the stable chimera state on the Poisson manifold}\label{sec:5}
The Poisson manifold can be identified with the hyperbolic disc $\mathbb{D}$. As we have seen, dynamics on the Poisson manifold in the stable chimera state are very simple: sub-population $A$ is synchronized, while centroid of sub-population $B$ evolves along a circle inside unit disc $\mathbb{D}$.
\begin{figure}[t]
\centering
  \begin{tabular}{@{}cc@{}}
    \includegraphics[width=.33\textwidth]{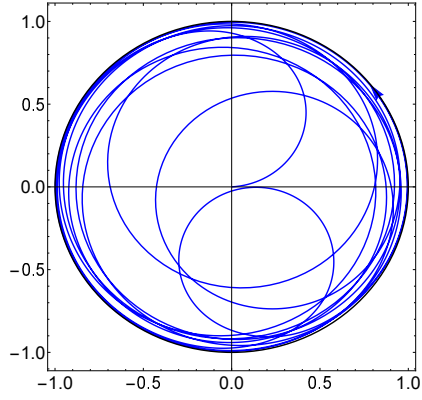}
  \end{tabular}
  \caption{\label{fig:6}M\" obius group action for desynchronized sub-population $B$ in breathing chimera: $\alpha_B(1000;t)$ at time interval $t \in [1000,1500]$. Initial distributions of oscillators are uniform. Notice that $\alpha_B(1000;t)$ does not return to zero.}
\end{figure}

However, simulations from the previous Section suggest that the corresponding trajectories on 6-dimensional manifold $G \times G$ are not that simple. Essentially, we have a projection map $\xi: G \times G \to \mathbb{D} \times \mathbb{D}$. This map has 2-dimensional fibers $S^1 \times S^1$ and dynamics on these fibers affect dynamics of oscillators, although this impact is not visible when observing only densities. These fibers come from decomposition $G \backslash S^1 = \mathbb{D}$ of the M\" obius group.

In this Section we conduct an analytic study of trajectories on $G \times G$ that are generated by system (\ref{chimera}). We start with the observation that for $t > T$ these dynamics are restricted to an invariant 3-torus ${\mathbb T}^3 \subset G \times G$. This torus is parametrized by angles $\Phi_A, \Phi_B$ and $\psi_B$. Since $r_A(t) = 1$ and $r_B(t) = const$ for $t>T$, we have that $\alpha_A(t) = e^{i \Phi_A(t)}$ and $\alpha_B(t) = r_B e^{i \varphi_B(t)}$. By plugging these expressions into (3) and discarding the ODE for $\psi_A$ as irrelevant, we obtain the dynamical system on ${\mathbb T}^3$
\begin{equation}
\label{3-torus}
\left\{
\begin{array}{lll}
\dot \Phi_A = f_A e^{i \Phi_A} + \omega + \bar f_A e^{-i \Phi_A};\\
\dot \Phi_B = r f_B e^{i \Phi_B} + \omega + \frac{1}{r} \bar f_B e^{-i \Phi_B};\\
\dot \psi_B = r f_B e^{i \Phi_B} + \omega + r \bar f_B e^{-i \Phi_B}.
\end{array}
\right.
\end{equation}
We will focus on the dynamics on an invariant 2-torus ${\mathbb T}^2 \subset G$ which is described by angles $\Phi_B$ and $\psi_B$ corresponding to the desynchronized sub-population.

We plug expressions for $f_A$ and $f_B$ into (\ref{3-torus}) and impose the condition that the right hand side of ODE for $\Phi_B$ must be real. This yields
\begin{equation}
\label{expr}
\nu \cos (\Phi_A - \Phi_B - \beta) + r \mu \cos \beta = 0.
\end{equation}
Notice that (\ref{expr}) is precisely the necessary condition for the stable chimera to exist in the model \cite{AMSW}. Then the system on ${\mathbb T}^2$ is rewritten as
\begin{equation}
\label{2-torus}
\left\{
\begin{array}{rcl}
\dot \Phi_B &=& - \frac{\mu}{2}(r^2 + 1) \sin \beta + \omega + (r + \frac{1}{r}) \frac{\nu}{2} \sin(\Phi_A - \Phi_B - \beta);\\
\dot \psi_B &=& - r^2 \mu \sin \beta + \omega + r \nu \sin(\Phi_A - \Phi_B - \beta).
\end{array}
\right.
\end{equation}
Substitution yields ODE for the difference between the two phases
$$
\frac{d}{dt}(\psi_B - \Phi_B) = (r - \frac{1}{r}) (\nu \sin(\Phi_A - \Phi_B - \beta) - r \mu \sin \beta).
$$
It follows from (\ref{expr}) that the difference $\Phi_A - \Phi_B$ is constant. Moreover, $\Phi_A = \Phi_B$ implies that $r_A = r_B = 1$. This corresponds to the full synchronization in the system. On the other hand, for the chimera state (i.e. when $r_B < 1$) one has $\Phi_A \neq \Phi_B$.

Then it follows from (\ref{2-torus}) that phases $\Phi_B$ and $\psi_B$ evolve with nonequal constant frequencies.

Further, from $\Phi_A - \Phi_B = const$ we have that right hand sides in (\ref{2-torus}) are $2 \pi$-periodic functions. Hence, $\Phi_B(t)$ and $\psi_B(t)$ are periodic.

We conclude that the stable chimera generates a quasiperiodic trajectory on the torus ${\mathbb T}^2$. In particular, this trajectory can be periodic if the ratio between frequencies of $\Phi_B$ and $\psi_B$ is a rational number. For simplicity set $\omega = 0$, then we get
$$
{\dot \Phi_B}/{\dot \psi_B} = \frac{1}{2} + \frac{1}{2 r_B^2}.
$$
The number $0 < r_B^2 < 1$ is not necessarily rational. Therefore, the trajectory on $G$ that corresponds to the stable chimera is not necessarily periodic. However, this subtle nuance is essentially irrelevant in numerical simulations. In this context we mention the Peixoto's theorem \cite{Peixoto} stating that quasiperiodic trajectories on 2-tori are structurally unstable. In other words, these trajectories can be destroyed by arbitrarily small perturbations in the system. Furthermore, when the quasiperiodic trajectory is destroyed, it is replaced by a stable periodic orbit, which is structurally stable.

\section{Simulations: M\" obius group actions on a generic 3-dimensional invariant submanifold}\label{sec:6}
In this Section we examine what is happening when we give up Assumption \ref{assum:1}. In other words, we will investigate chimeras that arise from system (\ref{chimera}) with densities that evolve off the Poisson manifold.

Underline that in this case the model is {\bf not solvable}. If the initial distributions are not uniform, system (\ref{r_B}) is not valid. Still, Proposition \ref{prop:1} remains valid regardless the initial distributions. Therefore, actions of the M\" obius group are still present and their actions on zero are still parameters $\alpha_A(t)$ and $\alpha_B(t)$ defined by (\ref{Mobius}) and evolving by (\ref{WS^2}). However, alpha's are no longer complex order parameters. This leaves us without strictly analytic tools and the results of this Section are strictly numerical.
\begin{figure}[t]
\centering
  \begin{tabular}{@{}cc@{}}
    \includegraphics[width=.38\textwidth]{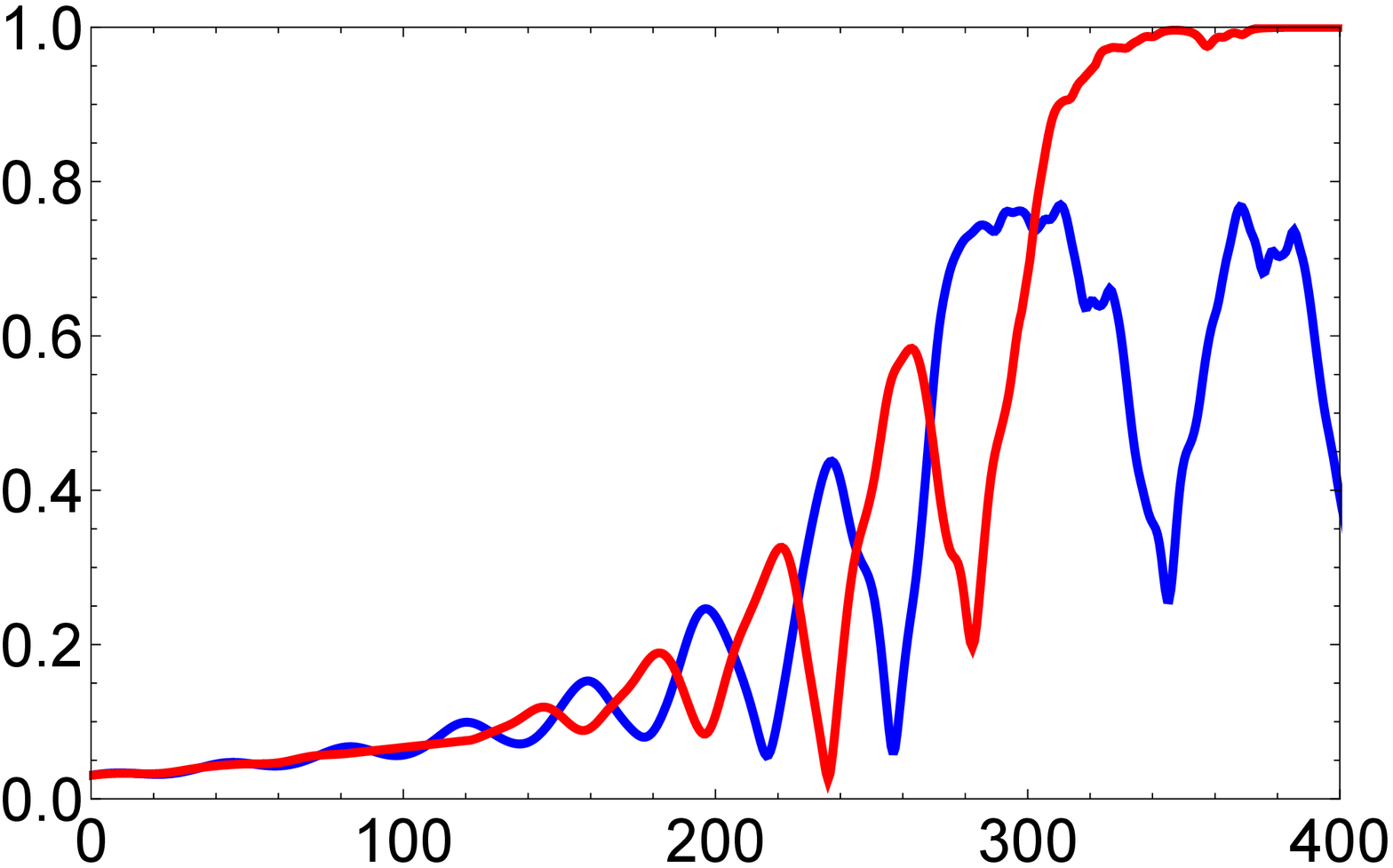}& \includegraphics[width=.383\textwidth]{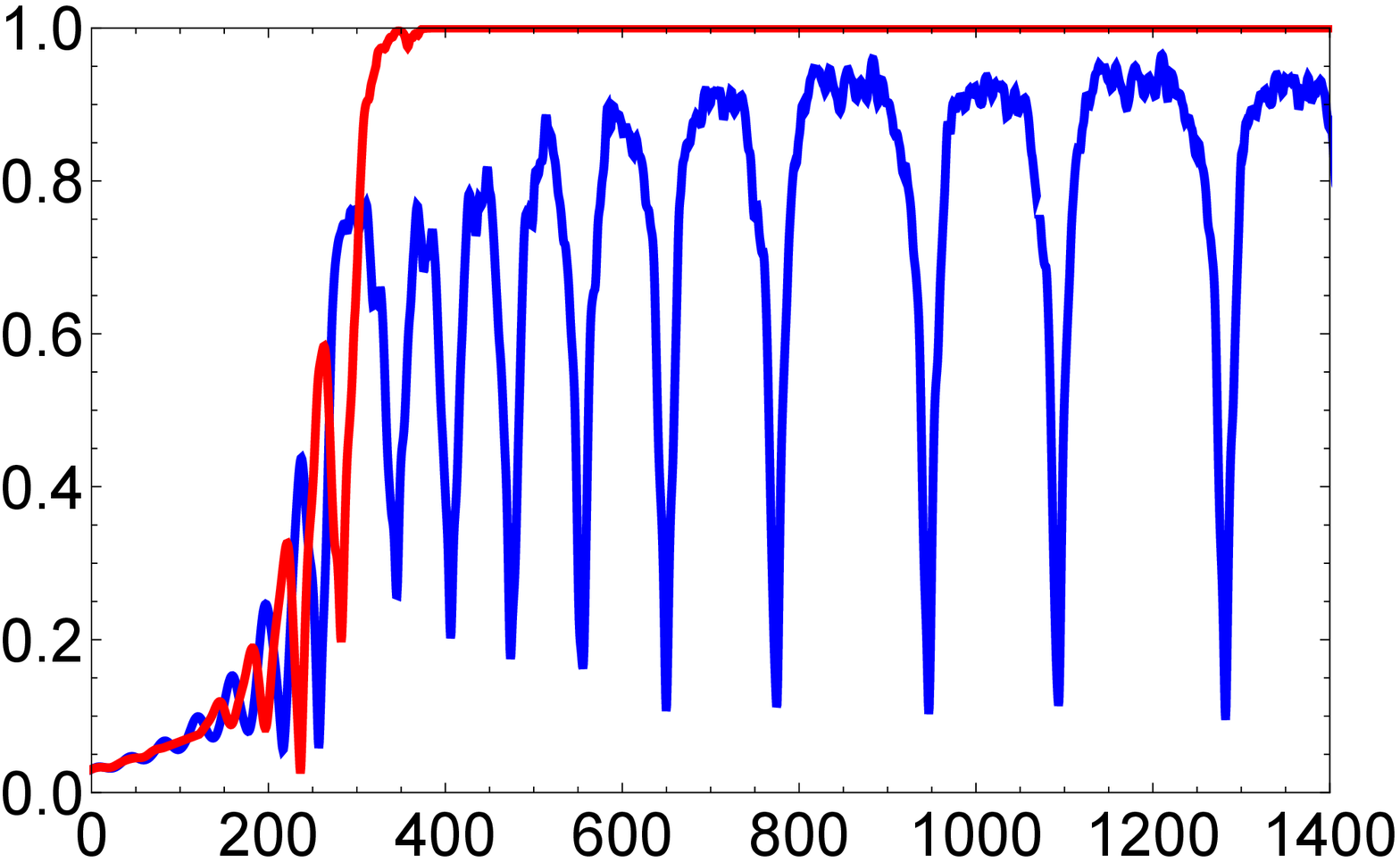}\\
\quad (a)&\quad (b)
  \end{tabular}
  \caption{\label{fig:7}Real order parameters $r_A(t)$ and $r_B(t)$ for sub-populations $A$ and $B$ in model (\ref{chimera}) with $\mu = 0.673, \nu = 0.327, \beta = \frac{\pi}{2} - 0.1$ (breathing chimera) and balanced non-uniform initial distributions on time intervals (a) $t \in [0,400]$ and (b) $t \in [0,1400]$.}
\end{figure}

For numerical simulations we need an example of a 3-dimensional manifold that is invariant for the evolution of Kuramoto oscillators. There are different examples of this kind. For instance, authors in \cite{EM2} used mixtures of Poisson kernels in order to study the dynamics off the Poisson manifold. Alternative examples of generic 3-dimensional invariant submanifolds can be borrowed from Directional Statistics: one can sample some commonly used probability distributions on $S^1$. (Emphasize, however, that there are only a few probability distributions on $S^1$ that are well studied and commonly used.) One possibility is to choose the von Mises distribution on $S^1$ as the starting point and consider submanifold in ${\cal P}(S^1)$ that consists of all probability measures that are obtained by the M\" obius group action on the von Mises distribution. This submanifold has been studied in detail in paper \cite{KJ} of Kato and Jones. We will refer to this invariant submanifold as {\it Kato-Jones family} (or Kato-Jones manifold). We refer to the paper \cite{KJ} for all details and exact formulas. Kato-Jones family has already been used in investigation of coupled oscillators in \cite{JC1}.

\begin{definition}
Probability measure on $S^1$ is called {\it balanced} if its mean value equals zero.
\end{definition}

We want to sample oscillators from a certain probability distribution that is balanced, but not uniform. To obtain such a distribution, we start by sampling random points $z_1,z_2,\dots,z_N$ on $S^1$ from the von Mises distribution $vM(\kappa)$ with fixed concentration parameter $\kappa = 2$. Then, there exists a unique (up to a rotation) M\" obius transformation that maps these points into a balanced set of points. (The previous assertion is the essence of so-called {\it Douady-Earle extension} \cite{DE}, more details in the next Section.) In order to compute this transformation, we solve the Kuramoto model with $N$ identical, globally coupled oscillators, with initial positions $z_1,z_2,\dots,z_N$ and {\bf repulsive} coupling. Then, real order parameter of the system will decrease and converge to zero (see \cite{Jac}). We run this Kuramoto model up to a certain moment $T$ when the order parameter equals zero. Due to \cite{MMS} the distribution obtained in such a way is a M\" obius transformation of the von Mises distribution $vM(2)$. In whole, we obtained a sample from a probability measure on $S^1$ that is: (a) balanced; (b) belongs to the Kato-Jones manifold. Underline that this measure is not uniform. In fact, it is a bimodal distribution with two peaks at antipodal points of $S^1$.
\begin{figure}[t]
\centering
  \begin{tabular}{@{}c@{}}
    \includegraphics[width=.33\textwidth]{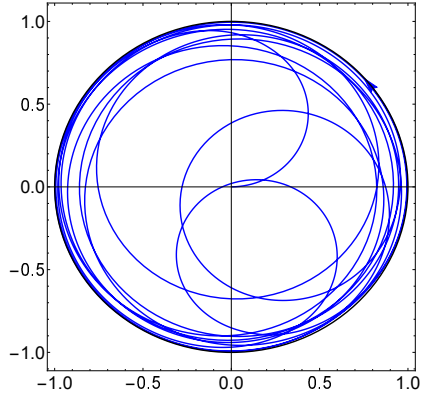}
  \end{tabular}
  \caption{\label{fig:8}M\" obius group action for desynchronized sub-population $B$ in breathing chimera: $\alpha_B(1000;t)$ at time interval $t \in [1000,1500]$. Initial distributions of oscillators are balanced, still, not uniform. Notice that $\alpha_B(1000;t)$ does not return to zero.}
\end{figure}
\begin{figure}[t]
\centering
  \begin{tabular}{@{}cc@{}}
    \includegraphics[width=.38\textwidth]{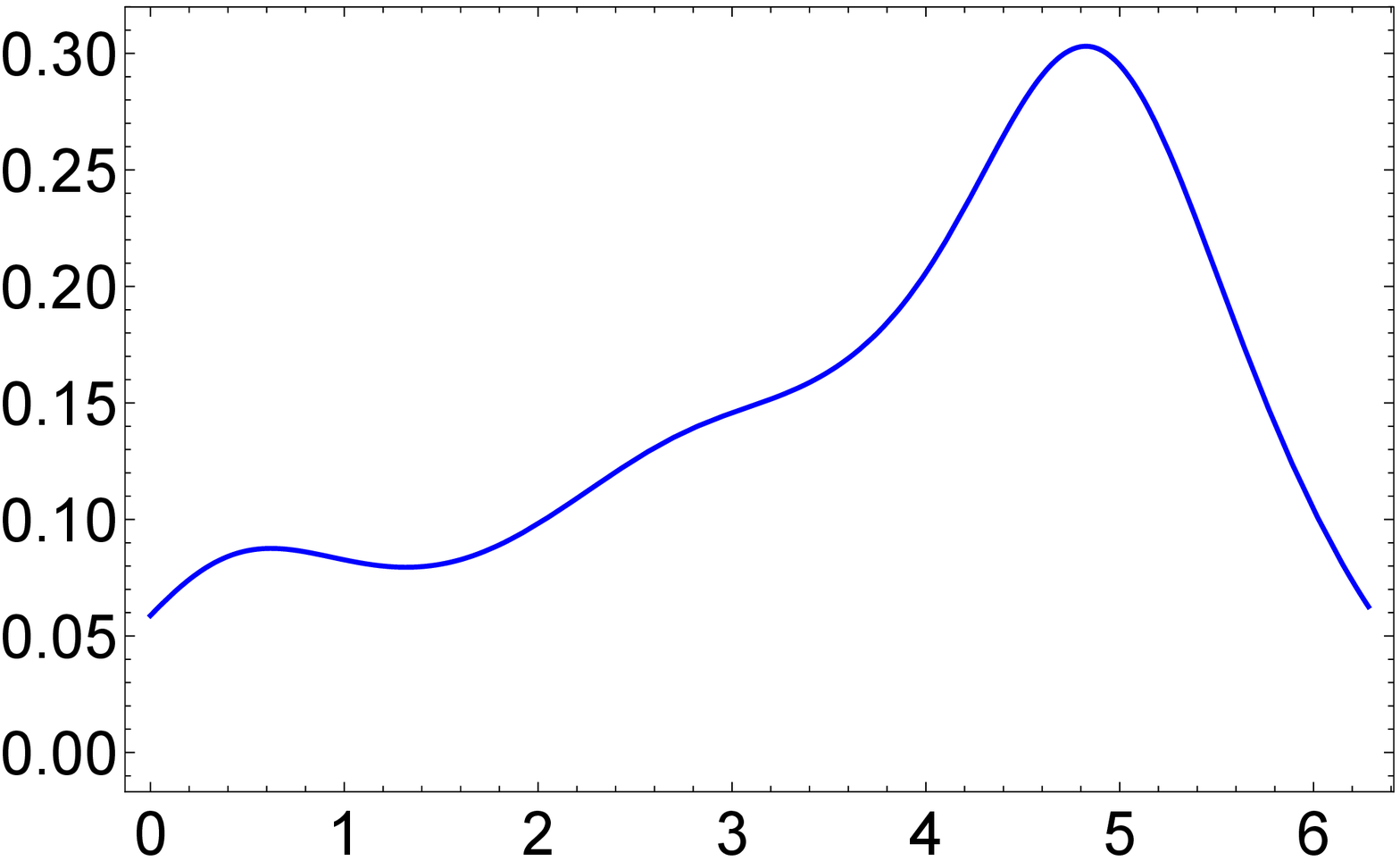}& \includegraphics[width=.38\textwidth]{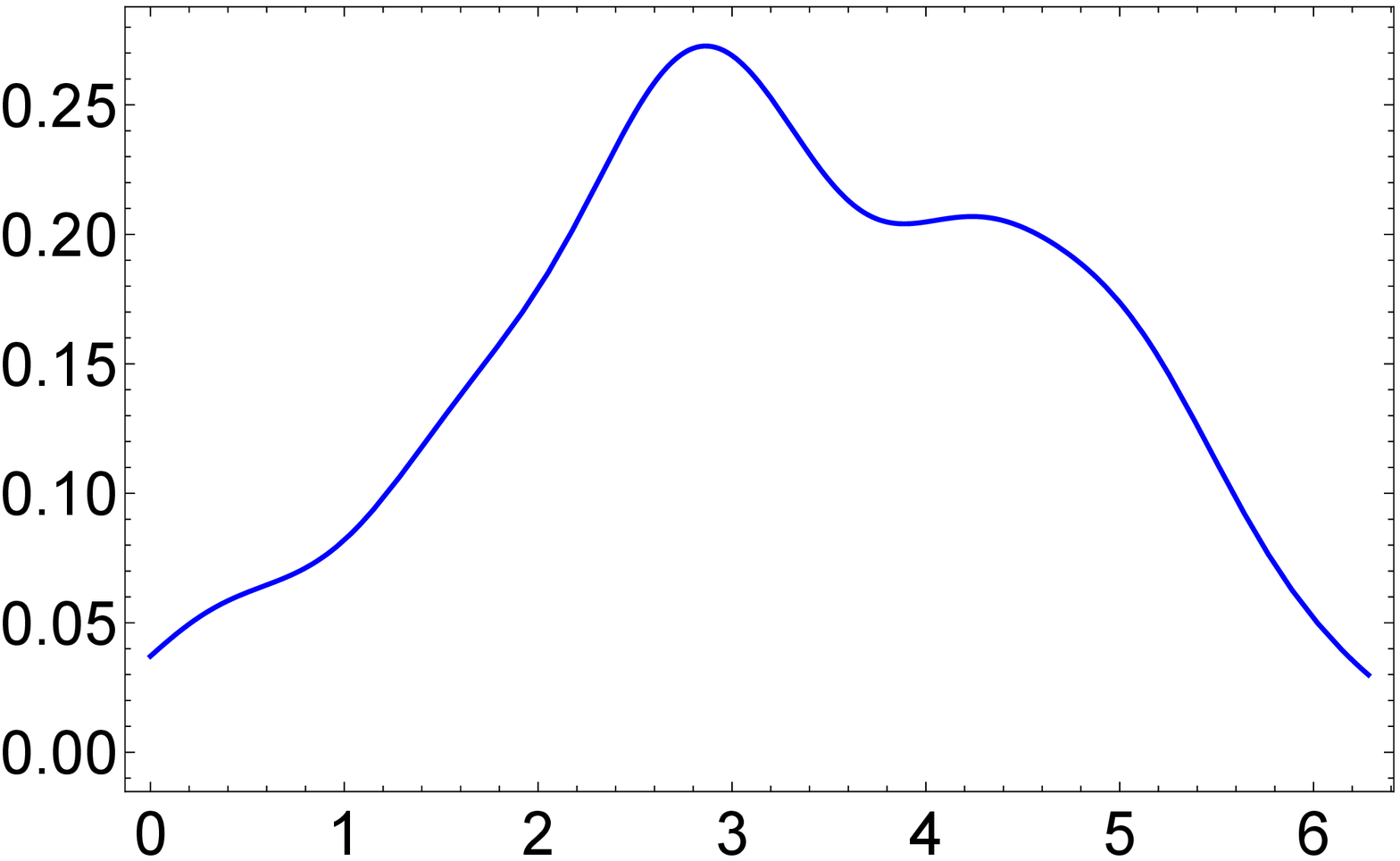}\\
\quad (a)&\quad (b)
  \end{tabular}
  \caption{\label{fig:9}Distributions of oscillators in desynchronized sub-population $B$ at moments $T_1$ and $T_2$, such that $r_B(T_1) = r_B(T_2)$: (a) $\rho_B(T_1)$ and (b) $\rho_B(T_2)$. Initial positions of oscillators are sampled from a balanced non-uniform distribution. Notice that distributions $\rho_B(T_1)$ and $\rho_B(T_2)$ {\bf are not} related by a simple rotation.}
\end{figure}

Further, we sample initial conditions for sub-populations $A$ and $B$ from this balanced Kato-Jones distribution and run system (\ref{chimera}). Simulation demonstrates that chimeras occur in such a setup, in a pretty much the same way as they occur on the Poisson manifold (although for slightly different parameter values), see Figure \ref{fig:7}. In Figure \ref{fig:8} we depict the projection of the corresponding trajectory from $G$ to $\mathbb{D}$. We can see that $\alpha_B(1000;t)$ does not turn to zero. This means that $g_{[1000,t]}^B$ is never a rotation. Underline that, there is nothing surprising in this conclusion in such a setup.

Figure \ref{fig:9} shows densities of the desynchronized sub-population at two moments $T_1$ and $T_2$, such that $r_B(T_1) = r_B(T_2)$. Both densities belong to the Kato-Jones manifold, however, they are obviously not related by a simple rotation. Of course, this invariant submanifold is larger than the Poisson manifold (its dimension is 3, not 2). Hence, it contains densities of various shapes with equal real (and complex) order parameters.

This makes a crucial difference with the evolution on the Poisson manifold: global phase $\psi_B$ affects the shape of density.

\section{Discussion: Difference between evolution on the Poisson manifold and on generic 3-dimensional invariant submanifolds}\label{sec:7}
Under Assumption \ref{assum:1}, the densities of oscillators' phases evolve on the Poisson manifold and exhibit very simple asymptotic behavior in the chimera state. However, the corresponding trajectory on $G \times G$ turns out to be much more complicated. In Section \ref{sec:5} we have investigated a projection map $\xi$ that maps the trajectory from 6-dimensional manifold $G \times G$ to 4-dimensional invariant manifold (product of two Poisson manifolds). This 4-dimensional manifold can be identified with $\mathbb{D} \times \mathbb{D}$. Map  $\xi : G \times G \to \mathbb{D} \times \mathbb{D}$ has fibers $S^1 \times S^1$. The impact of global phases $\psi_A \in S^1$ and $\psi_B \in S^1$ is hidden: they affect dynamics of individual oscillators, but do not affect the shape of densities.

In Section \ref{sec:6} we have studied evolution on generic 3-dimensional invariant submanifolds (off the Poisson manifold). In this case the projection map is one-to-one; it maps a trajectory from 6-dimensional manifold $G \times G$ to a trajectory on 6-dimensional invariant submanifold in ${\cal P}(S^1) \times {\cal P}(S^1)$. This projection does not have fibers. Impact of global phases is now clearly visible, they affect the shape of densities.

In order to shed some additional light on special properties of the Poisson manifold, we take a look from a slightly different point of view.

First of all, denote by $C(\mu)$ the mean value (centroid) of a measure $\mu$ on $S^1$. Of course, $C(\mu)$ is a point in the unit disc $\mathbb{D}$.

Second, in the proofs of propositions below we will also use the notion of {\it conformal barycenter} of a probability measure on $S^1$. We denote by $B(\mu)$ the conformal barycenter of a measure $\mu$. $B(\mu)$ is also a point in $\mathbb{D}$. For the definition of conformal barycenter we refer to the seminal paper \cite{DE} of Douady and Earle. \footnote{The notion of conformal barycenter appeared in recent studies of hyperbolic gradient flows generated by the Kuramoto model, see \cite{CEM,Jac}.}

We continue with some basic facts that are straightforward from the definition of conformal barycenter (see \cite{DE}).

\begin{proposition}
\label{prop:5}
Measure $\mu$ is balanced, if and only if $B(\mu) = C(\mu) = 0$.
\end{proposition}

\begin{proposition}
\label{prop:6}
Let $\mu$ and $\pi$ be two measures on $S^1$ related by a M\" obius transformation, i.e. $\mu = g_* \pi$ for some $g \in G$. Then $B(\mu) = g(B(\pi))$.
\end{proposition}

Proposition \ref{prop:6} claims that if a measure is transformed by a M\" obius transformation, then its conformal barycenter is also transformed by the same M\" obius transformation. This is an essential property of conformal barycenter. Notice that, in general, the same is not true for mean values (centroids) of measures on the circle.

\begin{proposition}
\label{prop:7}
Let $\mu \in {\cal P}(S^1)$ be an absolutely continuous measure. Then, there exists a unique (up to a rotation) M\" obius transformation $g \in G$, such that the measure $g_* \mu$ is balanced.
\end{proposition}

We now examine invariant submanifolds for actions of group $G$ in more detail.

\begin{corollary}
Invariant submanifold ${\cal M} \subset {\cal P}(S^1)$ for the action of $G$ is fully determined by the balanced distribution $\mu_0 \in {\cal M}$. In addition, each absolutely continuous measure $\mu \in {\cal P}(S^1)$ belongs to exactly one invariant submanifold in ${\cal P}(S^1)$.
\end{corollary}

\begin{proposition}
\label{prop:8}
Each invariant submanifold for the action of $G$ contains a unique (up to a rotation) balanced distribution and (as the limit case) the delta distribution.
\end{proposition}

\begin{proof}
The existence of a balanced distribution follows from Proposition \ref{prop:7}. In order to prove that each invariant submanifold also contains (as the limit case) the delta distribution, one might use results of \cite{EM1} about attracting sets for the Kuramoto model. We will take a different way. Let ${\cal M} \subset {\cal P}(S^1)$ be an invariant submanifold for $G$. Due to Proposition \ref{prop:7}, there is a balanced measure $\mu_0 \in {\cal M}$. Next, pick an arbitrary sequence of points $z_1, z_2, \dots$ in the unit disc, such that $|z_j| \to 1$ as $j \to \infty$. In other words, sequence $z_j$ converges towards the boundary (infinite horizon) of the unit disc. Obviously, there exists a sequence $g_1, g_2, \dots$ of M\" obius transformations that map zero to these points: $z_1 = g_1(0), z_2 = g_2(0), \dots$  Due to Proposition \ref{prop:6}, $z_j$ are conformal barycenters of measures $\mu_j = {g_j}_* \, \mu_0$.

We now have a sequence of measures $\mu_1, \mu_2, \dots$ with conformal barycenters $z_1 = B(\mu_1), z_2 = B(\mu_2), \dots$, such that $|B(\mu_j)| \to 1$. We need to show that their mean values $C(\mu_j)$ satisfy the same. This follows from the relation between conformal barycenter and mean value that has been derived in \cite{MMS}
\begin{equation}
\label{centrbary}
B(\mu_j) = C(\mu_j) + (1 - |B(\mu_j)|^2) \sum \limits_{n=1}^\infty (-1)^n \bar c_n e^{i n \psi} \bar B(\mu_j)^{n-1}.
\end{equation}
Here, $c_n$ are Fourier coefficients of the balanced measure $\mu_0$:
$$
d \mu_0(\theta) = \frac{1}{2 \pi} \sum \limits_{n=-\infty}^\infty c_n e^{i n \theta} d \theta, \quad c_{-n} = c_n, \; c_0 =1.
$$
From (\ref{centrbary}) it is easy to see that if $|B(\mu_j)| \to 1$, then $|C(\mu_j)| \to 1$.

Hence, moduli of mean values (centroids) of measures $\mu_j = {g_j}_* \, \mu_0$ converge towards 1. This implies that $\mu_j$ converge (weakly) towards the delta distribution.
\end{proof}

\begin{remark}
Proposition \ref{prop:8} asserts that each invariant submanifold contains a unique (up to a rotation) balanced measure (in the context of couple oscillators, this is usually called {\it a fully incoherent state}) and the delta measure (that corresponds to the coherent state in ensembles of coupled oscillators). The balanced measure in the Poisson manifold is uniform (unique and invariant w.r.t. planar rotations). Balanced measures in other invariant manifolds are multimodal.
\end{remark}

\begin{proposition}
\label{prop:9}
Let $\mu \in {\cal P}(S^1)$ be an absolutely continuous measure that is {\bf not} balanced. Suppose that $\mu = g_* \mu$ for some $g \in G$ that is {\bf not} the identity transformation. Then, $\mu$ is the Poisson kernel.
\end{proposition}

\begin{proof}
We have that $\mu = g_* \mu$. Since $\mu$ is not balanced, and not a Poisson kernel, its conformal barycenter and mean value are two distinct points in $\mathbb{D}$, see \cite{DE}. Hence, $g$ maps $\mathbb{D}$ into itself and fixes two distinct points in $\mathbb{D}$: $g(\mathbb{D}) = \mathbb{D}, \, g(B(\mu)) = B(\mu), \, g(C(\mu)) = C(\mu)$. Identity is the only transformation that meets all these conditions.
\end{proof}

\begin{remark}
Equation (\ref{centrbary}) shows that conformal barycenter and mean value are two distinct points, except in the following two cases: (a) $c_n = 0$ (in this case measure $\mu_j$ is the Poisson kernel, since the corresponding balanced measure $\mu_0$ is uniform); and (b) $B(\mu_j) = 0$ (in this case $\mu_j$ is balanced, since $B(\mu_j)$ is its conformal barycenter). In all other cases the series on the right hand side of (\ref{centrbary}) are distinct from zero.
\end{remark}

The above Proposition emphasizes one exceptional feature of Poisson kernels: their mean values coincide with conformal barycenters. The parameter $\alpha_l = r_l e^{i \Phi_l}$ is both conformal barycenter and the mean value of density (\ref{Poisson}). Hence, any disc-preserving M\" obius transformation that fixes $\alpha_l$ maps the Poisson kernel into itself. Obviously, identity is not the only transformation in $G$ that fixes $\alpha_l$.

\section{Conclusion}\label{sec:8}
This study is based on the observation that the solvable chimera model (\ref{chimera}) induces two coupled actions of the M\" obius group on the unit disc. The corresponding dynamics on the group $G \times G$ is given by two coupled Watanabe-Strogatz systems (\ref{WS^2}). This observation has been first exploited in \cite{PR}.

Abrams et al. have shown in \cite{AMSW} that (\ref{chimera}) exhibits the stable chimera state for certain values of parameters. This state is sometimes also named a stationary chimera \cite{BSOP}. Recent findings of Engelbrecht and Mirollo (\cite{EM2}) demonstrate that stability of this chimera state is deceptive. Here, we have shown that the "stationarity" of this chimera state is also deceptive.

Our study unveiled an unexpected effect in the solvable chimera model: the density of oscillators is stationary (in the rotating coordinate frame), while individual oscillators are not stationary.

Proposition \ref{prop:9} demonstrates that this "deceptive stationary state" is possible only on the 2-dimensional Poisson manifold. This sheds an additional light on paper \cite{EM2} which demonstrated that dynamics near the Poisson manifold are more complicated than one might expect. In particular, authors in \cite{EM2} provide an explicit example of how a stable stationary state on the Poisson manifold bifurcates into limit cycles under a small perturbation off the Poisson manifold.

We have further pointed out that global phases $\psi_A$ and $\psi_B$ act as "hidden variables" whose impact is not visible on the shape of densities. In fact, chimera corresponds to a quasiperiodic trajectory on group $G \times G$ and the full picture of the dynamics can be restored only when observing this trajectory.

On the other side, when the evolution takes a place on generic 3-dimensional invariant submanifolds, global phases clearly affect the shape of densities.

We conclude this paper with a brief remark about some analogous concepts from Mathematical Physics.

\subsection{Broader significance}

There is a strong analogy between the Poisson manifold in our model and $SU(1,1)$-coherent states in quantum theories, notably in Quantum Optics. \footnote{In order to avoid confusion, underline that the term "coherent state" in quantum theories has nothing in common with terms "incoherent state" and "fully coherent state" in context of coupled oscillators.} Indeed, Poisson kernels (seen as analytic functions on the unit disc) fit into general Perelomov's framework of coherent states  \cite{Gazeau,Perelomov}.

In our case, there is group $G$ acting on the space ${\cal P}(S^1)$. The maximal compact subgroup of $G$ is the circle $S^1$ which can be identified with group $SO(2)$ of planar rotations. Furthermore, factorizing $G$ w.r.t. $S^1$ yields decomposition $G \setminus S^1 = \mathbb{D}$, where $\mathbb{D}$ is hyperbolic disc. The state in ${\cal P}(S^1)$ which is invariant w.r.t. the action of this maximal compact subgroup (group of planar rotations) is obviously the uniform measure on $S^1$. Hence, uniform measure corresponds to the so-called "ground state" in quantum theories. Then the manifold of coherent states is obtained by actions of $G$ on this ground state. In our case, this is the Poisson manifold, as stated in Proposition \ref{prop:2}. Moreover, the special feature of coherent states is that each of them is uniquely labeled by a point in $\mathbb{D}$. Thus, evolution of the quantum system on the invariant submanifold of coherent states can be restricted to $\mathbb{D}$.

In fact, Poisson kernels have already appeared as $SU(1,1)$-coherent states in a number of papers \cite{DSV,Kisil,Vourdas1,Wunsche}. In such a setup, quantum states are represented by analytic functions on the unit disc \cite{Vourdas2}. Then, square roots of Poisson kernels appear as coherent states (eigenfunctions of the hyperbolic Laplace-Beltrami operator). These are complex-valued analytic functions $f(z)$ in $\mathbb{D}$. Then, there exists a limit function $f(z) \to g(\varphi)$ as $|z| \to 1$ that is defined on $S^1$. The square of the modulus of this limit function $\rho(\varphi) = |g(\varphi)|^2$ is precisely the Poisson kernel which is interpreted as a phase distribution of a coherent quantum state. Therefore, the Poisson manifold is the manifold of phase distributions of $SU(1,1)$-coherent states.

The present paper demonstrates how fibers affect the dynamics on the manifold of coherent states. Although $SU(1,1)$-coherent states are uniquely labeled by points in $\mathbb{D}$, the trajectory on $G$ turns out to be much subtler than its projection on $\mathbb{D}$.

Recent findings of \cite{EM2} suggest that the manifold of $SU(1,1)$-coherent states is not, in general, attractive for quantum evolution on the full state space.

This indicates that something similar might also happen on other kinds of coherent states, for instance $SU(2)$-coherent states, or coherent states for the Weyl-Heisenberg group.

To conclude, chimera model (\ref{chimera}) induces intriguing dynamics on $G \times G$ and the group-theoretic approach facilitates understanding collective motions in the desynchronized sub-population.
This approach unveils some hidden subtleties.






\end{document}